\documentclass[USenglish,nolineno]{socg-lipics-v2019}
\usepackage[utf8]{inputenc}
\usepackage{subcaption}

\newcommand{\eps}{\varepsilon}

\newcommand{\IGNORE}[1]{}

\newcommand{\cC}{{\cal C}}
\newcommand{\cD}{{\cal D}}

\newcommand{\cA}{{\cal A}}
\newcommand{\cB}{{\cal B}}

\newcommand{\cont}{\mathtt{cont}}
\newcommand{\down}{\mathtt{down}}
\newcommand{\halfr}[1]{(-\infty,#1]}
\newcommand{\ra}{\mathrm{rank}}
\newcommand{\pred}{\mathrm{pred}}
\newcommand{\ssucc}{\mathrm{succ}}

\EventEditors{Sergio Cabello and Danny Z. Chen}
\EventNoEds{2}
\EventLongTitle{36th International Symposium on Computational Geometry (SoCG 2020)}
\EventShortTitle{SoCG 2020}
\EventAcronym{SoCG}
\EventYear{2020}
\EventDate{June 23--26, 2020}
\EventLocation{Z\"{u}rich, Switzerland}
\EventLogo{socg-logo}
\SeriesVolume{164}
   
\title{Four-Dimensional Dominance Range Reporting in Linear Space}

\author{Yakov Nekrich}{Department of Computer Science, Michigan Technological University}{yakov.nekrich@googlemail.com}{}{}

\authorrunning{Y. Nekrich}
\Copyright{Yakov Nekrich}



\hideLIPIcs  

\begin{document}

\maketitle


\begin{abstract}
  In this paper we study the four-dimensional dominance range reporting problem  and present data structures with linear or almost-linear space usage. Our results can be also used to answer four-dimensional queries that are bounded on five sides. The first data structure presented in this paper uses linear space and answers queries in 
 $O(\log^{1+\eps}n + k\log^{\eps} n)$ time, where $k$ is the number of reported points, $n$ is the number of points in the data structure, and $\eps$ is an arbitrarily small positive constant. Our second data structure uses $O(n \log^{\eps} n)$ space and answers queries in $O(\log n+k)$ time.

These are the first data structures for this problem that use linear (resp.\ $O(n\log^{\eps} n)$) space and answer queries in poly-logarithmic time. For comparison the fastest previously known linear-space or $O(n\log^{\eps} n)$-space data structure supports queries in $O(n^{\eps} + k)$ time (Bentley and Mauer, 1980).  Our results can be  generalized to $d\ge 4$ dimensions. For example, we can answer $d$-dimensional dominance range reporting queries in $O(\log\log n (\log n/\log\log n)^{d-3} + k)$ time using $O(n\log^{d-4+\eps}n)$ space. Compared to the fastest previously known result~(Chan, 2013), our data structure reduces the space usage by $O(\log n)$ without increasing the query time.
\end{abstract}
\section{Introduction}
In the orthogonal range searching problem we keep a set $S$ of multi-dimensional points in a data structure so that for an arbitrary axis-parallel query rectangle $Q$ some information about points in $Q\cap S$ can be computed efficiently. Range searching is one of the most fundamental and widely studied problems in computational geometry.  Typically we want to  compute some aggregate function on $Q\cap S$ (range aggregate queries), generate the list of points in $S$ (reporting queries) or determine whether $Q\cap S=\emptyset$ (emptiness queries). In this paper we study the  complexity of  four-dimensional orthogonal range reporting and orthogonal range emptiness queries in the case of dominance queries and in the case when the query range is bounded on  five sides. We demonstrate for the first time that in this scenario both queries can be answered in poly-logarithmic time using linear or almost-linear space. 

Range trees, introduced by Lueker~\cite{Lueker78} in 1978 and Bentley~\cite{Bentley80} in 1980, provide a solution for the range reporting problem in $O(n\log^dn)$ space and $O(\log^d n+k)$ time for any constant dimension $d$. Henceforth $k$ denotes the number of points in the answer to a reporting query and $n$ denotes the number of points in the data structure. A number of improvements both in time and in space complexity were obtained in the following decades. See e.g.,~\cite{McCreight85,Chazelle86,ChazelleG86,ChazelleG86a,ChazelleE87,Chazelle88,Overmars88,Chazelle90a,Chazelle90b,SubramanianR95,VengroffV96,AlstrupBR00,AlstrupBR01,Nekrich07isaac,Nekrich07,Nekrich07algorithmica,KarpinskiN09,NN12,ChanLP11,Chan13} for a selection of previous works on range reporting and related problems. Surveys of previous results can be found in~\cite{Agarwal04,agarwal1999geometric,Nekrich16}.  We say that a range query is $f$-sided if the query range is bounded on $f$ sides, i.e., a query can be specified with $f$ inequalities; see Fig.~\ref{fig:queriesex} on p.~\pageref{fig:queriesex}. Researchers noticed that the space and time complexity of range reporting depends not only on the dimensionality: the number of sides that bound the query range is also important.  Priority search tree, introduced by McCreight~\cite{McCreight85}, provides an $O(n)$  space and $O(\log n + k)$ time solution for 3-sided range reporting queries in two dimensions. In~\cite{ChazelleE87} Chazelle and Edelsbrunner have demonstrated that three-dimensional $3$-sided queries (aka   three-dimensional dominance queries) can be answered in $O(\log^2 n+k)$ time using an $O(n)$ space data structure. In 1985 Chazelle~\cite{Chazelle88} described a compact version of the two-dimensional range tree that uses $O(n)$ space and supports general (4-sided) two-dimensional range reporting queries in $O(\log n+ k\log^{\eps} n)$ time, where $\eps$ denotes an arbitrarily small positive constant. In~\cite{Chazelle88} the author also presented an $O(n)$ space data structure that supports 5-sided three-dimensional reporting queries in  $O(\log^{1+\eps}n+k\log^{\eps}n)$ time. Bentley and Mauer~\cite{BentleyM80} described a linear-space data structure that supports $d$-dimensional range reporting queries for any constant $d$; however, their data structure has prohibitive query cost $O(n^{\eps} + k)$. 

Summing up, we can answer range reporting queries in poly-logarithmic time using an $O(n)$ space data structure 
when the query is bounded on at most 5 sides and the query is in two or three dimensions.   
Significant improvements were achieved on the query complexity of this problem in each case; see Table~\ref{tbl:linspace}. However, surprisingly,  linear-space and polylog-time data structures are known only for the above mentioned special cases of the range reporting.  For example, to answer four-dimensional 4-sided queries (four-dimensional dominance queries) in polylogarithmic time  using previously known solutions one would need $\Omega(n\log n/\log\log n)$ space. This situation does not change when we increase the space usage to $O(n\log^{\eps}n)$ words: data structures with poly-logarithmic time are known for the above described  special cases only. See Table~\ref{tbl:almostlin}.

\begin{figure}[tb]
  \centering
  \begin{tabular}[tbh]{ccc} \includegraphics[width=.3\textwidth,page=2]{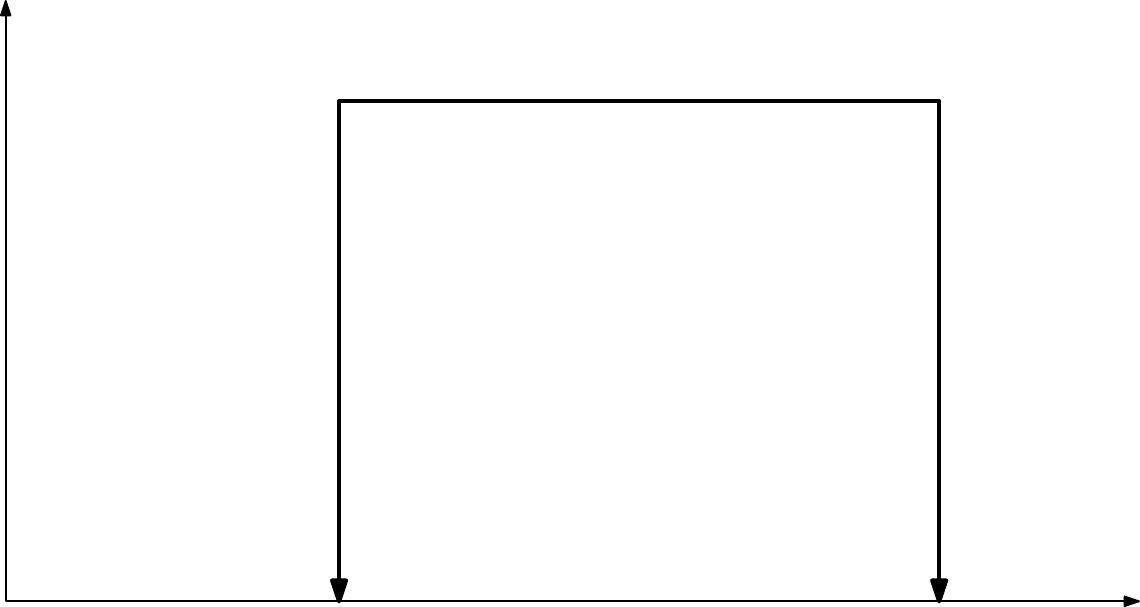} &  \hspace*{.1\textwidth}\includegraphics[width=.15\textwidth,page=3]{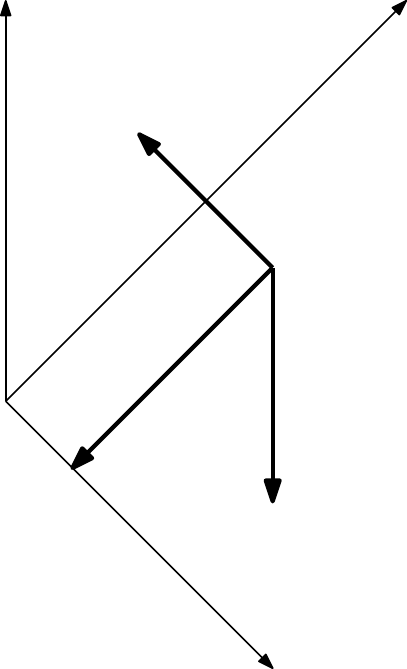} & \hspace*{.15\textwidth}\includegraphics[width=.2\textwidth,page=4]{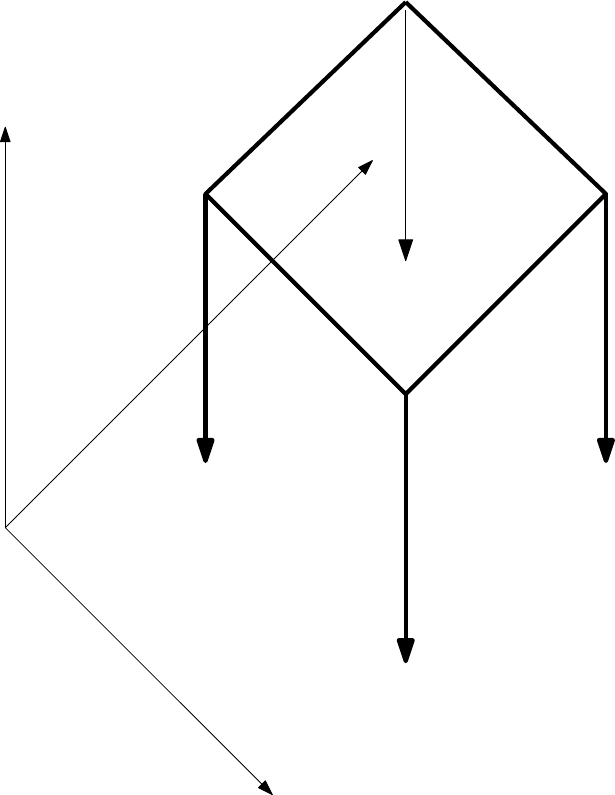} \\
(a) & (b) & (c) \\
  \end{tabular}
 
  \caption{Examples of  queries in two and three dimensions. (a) A two-dimensional 3-sided query (b) A three-dimensional 3-sided (dominance) query (c) A three-dimensional 5-sided query.}
  \label{fig:queriesex}
\end{figure}

Previous results raise the question about  low-dimensional range reporting. What determines the  complexity of range reporting data structures in $d\ge 4$ dimensions: the dimensionality or the number of sides in the query range? The lower bound of Patrascu~\cite{Patrascu11} resolves this question with respect to query complexity. It is shown in~\cite{Patrascu11} that any data structure using $O(n\mathtt{polylog}(n))$ space needs $\Omega(\log n/\log\log n)$ time to answer four-dimensional dominance (4-sided) queries.  On the other hand, two- and three-dimensional 4-sided queries can be answered in $O(\log\log n + k)$ time using $O(n\mathtt{polylog}(n))$ space. In this paper we address  the same question with respect to space complexity. 

We demonstrate that four-dimensional 5-sided queries can be answered in $O(\log^{1+\eps} n+k\log^{\eps} n)$ time using an $O(n)$ space data structure. Our data structure can also support 5-sided emptiness queries in $O(\log^{1+\eps}n)$ time. If the space usage is slightly increased to $O(n\log^{\eps} n)$, then we can answer reporting and emptiness queries in $O(\log n +k)$ and $O(\log n)$ time respectively. 
For comparison, the fastest previous method~\cite{Chan11} requires $O(n\log^{1+\eps} n)$ space and supports queries in $O(\log n + k)$ time. Since dominance queries are a special case of 5-sided queries, our results can be used to answer four-dimensional dominance queries within the same time and space bounds.  Using standard techniques, our results can be generalized to $d$ dimensions for any constant $d\ge 4$: We can answer $d$-dimensional dominance queries in $O(\log\log n (\log n/\log\log n)^{d-3} + k)$ time and $O(n \log^{d-4+\eps} n)$ space. We can also answer arbitrary $(2d-3)$-sided $d$-dimensional range reporting queries within the same time and space bounds.


Our base data structure is the range tree on the fourth coordinate and every tree node contains a data structure that answers three-dimensional dominance queries. We design a space-efficient representation of points stored in tree nodes, so that each point uses only $O(\log\log n)$ bits.  Since each point is stored $O(\log n/\log\log n)$ times in our range tree, the total space usage is $O(n\log n)$ bits. Using our representation we can answer three-dimensional queries on tree nodes  in $O(\log^{\eps}n)$ time; we can also ``decode'' each point, i.e., obtain the actual point coordinates from its compact representation, in $O(\log^{\eps} n)$ time. Efficient representation of points is the core idea of our method: we keep a geometric construct called \emph{shallow cutting} in each tree node and exploit the relationship between shallow cuttings in different nodes. Shallow cuttings were extensively used in previous works to decrease the query time.  But, to the best of our knowledge, this paper is the first that uses shallow cuttings to reduce the space usage. The novel part of our construction is a combination of several shallow cuttings that allows us to navigate between the nodes of the range tree and ``decode'' points from their compact representations.  We recall standard  techniques used by our data structure in Section~\ref{sec:prelim}. The linear-space data structure is described in Section~\ref{sec:5sid4dim}. In Section~\ref{sec:decodingfast} we show how the decoding cost can be reduced to $O(1)$  by slightly increasing the space usage. The data structure described in  Section~\ref{sec:decodingfast} supports queries in $O(\log n +k)$ time and uses $O(n\log^{\eps} n)$ space.  Previous and new results for 4-sided queries in four dimensions are listed in Table~\ref{tbl:dom4dim}. Compared to the only 
previous linear-space data structure~\cite{BentleyM80}, we achieve exponential speed-up in query time. Compared to the fastest previous result~\cite{Chan13}, our data structure reduces the space usage by $O(\log n)$ without increasing the query time.


The model of computation used in this paper is the standard RAM model. 
The space is measured in words of $\log n$ bits and we can perform standard arithmetic operations, such as addition and subtraction, in $O(1)$ time. Our data structures rely on bit operations, such as bitwise AND or bit shifts or identifying the most significant bit in a word. However such operations are performed only on small integer values (i.e., on positive integers bounded by $n$)  and  can be easily implemented using  look-up tables. Thus our data structures can be implemented on the RAM model with standard arithmetic operations and arrays. 

We will assume in the rest of this paper that all point coordinates are bounded by $n$. The general case can be reduced to this special case using the reduction to rank space, described in Section~\ref{sec:prelim}. All results of this paper  remain valid when the original point coordinates are real numbers.



\begin{table}[tb]
  \centering
  \begin{tabular}[tb]{|l|c|c|c|c|}
    \hline
    Query   & \multicolumn{2}{c|}{First}  & \multicolumn{2}{c|}{Best} \\ \cline{2-5}
    Type         &  Ref  & Query Time & Ref & Query Time \\
    \hline
    2-D 3-sided &  \cite{McCreight85}, 1985 & $O(\log n+k)$ & \cite{AlstrupBR00}, 2000 & $O(k+1)$ \\
    2-D 4-sided &  \cite{Chazelle88}, 1985 & $O(\log n + k\log^{\eps} n)$ & \cite{ChanLP11}, 2011 & $O((k+1)\log^{\eps} n)$\\
    3-D 3-sided &  \cite{ChazelleE87}, 1986    & $O(\log^2 n+k)$ & \cite{Chan11}, 2011 & $O(\log\log n+k)$ \\
    3-D 5-sided &  \cite{Chazelle88}, 1985   & $O(\log^{1+\eps} n + k\log^{\eps} n)$ &  \cite{FarzanMR12}, 2012 &  $O(\log n\log\log n)$\textsuperscript{{\small E}} \\ \hline
    4-D 5-sided & New   & \multicolumn{3}{c|}{$O(\log^{1+\eps} n + k\log^{\eps}n)$}\\\hline
  \end{tabular}
  \caption{Linear-space data structures for different types of queries. The second column provides the reference to the first data structure achieving linear space and the year of the first publication. The third column contains the query time of the first data structure. The fourth and the fifth columns contain the same information about  the best (fastest) currently known data structure. Result marked with {\small E} supports only emptiness queries. Data structures in rows 4 and 5 also support $4$-sided queries.}
\label{tbl:linspace}
\end{table}

\begin{table}[tb]
  \centering
  \begin{tabular}[tb]{|l|c|c|c|c|}
    \hline
    Query   & \multicolumn{2}{c|}{First}  & \multicolumn{2}{c|}{Best} \\ \cline{2-5}
    Type         &  Ref  & Query Time & Ref & Query Time \\
    \hline
    2-D 4-sided &  \cite{Chazelle88}, 1985  & $O(\log n +k)$ & \cite{AlstrupBR00}, 2000 & $O(\log \log n +k)$\\
    3-D 5-sided &  \cite{Chazelle88}, 1985   & $O(\log n + k)$ &  \cite{ChanLP11}, 2011 &  $O(\log\log n+k)$\\ \hline
    4-D 5-sided & New   & \multicolumn{3}{c|}{$O(\log  n + k)$}   \\\hline
  \end{tabular}
  \caption{$O(n\log^{\eps} n)$-space data structures for orthogonal range reporting. The second and the third columns contain the reference to and the query time of the first data structure. The fourth and the fifth columns contain the reference to and the query time of the best (fastest) currently known data structure. Data structures in rows 2 and 3 also support $4$-sided queries.}
\label{tbl:almostlin}
\end{table}

\begin{table}[tb]
  \centering
  \begin{tabular}{|c|c|c|} \hline
    Ref.\ & Space & Query Time \\ \hline
    \cite{BentleyM80}  & $O(n)$       & $O(n^{\eps} +k)$\\
    \cite{ChazelleE87} & $O(n\log n)$ & $O(\log^2 n + k)$\\
    \cite{JaJaMS04} & $O(n\log n)$ & $O(\log^2 n/\log \log n+ k)$\\
    \cite{Nekrich07}  & $O(n\log^{2+\eps} n)$ & $O(\log n\log\log n + k)$ \\
    \cite{Afshani08} & $O(n\log^{1+\eps} n)$ & $O(\log n\log\log n + k)$\\
    \cite{Chan11} & $O(n\log^{1+\eps} n)$ & $O(\log n + k)$\\
    \cite{Chan11} & $O(n \log n)$ & $O(\log n \log\log n + k)$\\ \hline
    New     & $O(n\log^{\eps} n)$ & $O(\log n + k)$\\
    New     & $O(n)$ & $O(\log^{1+\eps} n + k\log^{\eps} n)$ \\ \hline
  \end{tabular}
  \caption{Previous and new results on dominance range reporting in four dimensions. \IGNORE{We remark that our new data structures can also answer 5-sided queries.} Results in lines 2, 3, and 7 can be modified so that space is decreased to $O(n\log n/\log \log n)$ and the query time is increased by $O(\log^{\eps}n)$ factor.}
  \label{tbl:dom4dim}
\end{table}


\section{Preliminaries}
\label{sec:prelim}
In this paper $\eps$ will denote an arbitrarily small positive constant. We will consider four-dimensional points in a space with coordinate axes denoted by $x$, $y$, $z$, and $z'$. The following techniques belong to the standard repertoire of  geometric data structures.  

\subparagraph*{Shallow Cuttings.} A point $q$ dominates a point $p$ if and only if every coordinate of $q$ is larger than or equal to the corresponding coordinate of $p$. The \emph{level} of a point $q$ in a set $S$ is the number of points $p$ in $S$, such that $q$ dominates $p$ (the point $q$ is not necessarily in $S$). We will say that a \emph{cell} $C$ is a region of space dominated by a point $q_c$ and we will call $q_c$ the apex point of $C$. A $t$-shallow cutting of a set $S$ is a collection of cells $\cC$, such that (i) every point in $\mathbb{R}^d$ with level at most $t$ (with respect to $S$) is contained in some cell $C_i$ of $\cC$ and (ii) if a point $p$ is contained in some cell $C_j$ of $\cC$, then the level of $p$ in $S$ does not exceed $2t$. The size of a shallow cutting is the number of its cells. We can uniquely identify a shallow cutting $\cC$ by listing its cells and every cell can be identified by its apex point. Since the level of any point in a cell $C_j$ does not exceed $2t$, every cell $C_j$ contains at most $2t$ points from $S$, $|C_j\cap S|\le 2t$ for any $C_j$ in $\cC$. 

There exists a $t$-shallow cutting of size $O(n/t)$ for $d=2$~\cite{VengroffV96} or $d=3$ dimensions~\cite{Afshani08}. 
Shallow cuttings and related concepts are frequently used in data structures  for three-dimensional dominance range reporting queries.  

Consider a three-dimensional point $q_3=(a,b,c)$ and the corresponding dominance range $Q_3=\halfr{a}\times \halfr{b}\times \halfr{c}$. We can find a cell $C$ of a $t$-shallow cutting $\cC$ that contains $q_3$ (or report that there is no such $C$) by answering a point location query in a planar rectangular subdivision of size $O(n/t)$.  Point location queries in a rectangular subdivision of size $n$ can be answered in $O(\log\log n)$ time using an $O(n)$-space data structure~\cite{Chan13}.

\subparagraph*{Range Trees.}
A range tree is a data structure that  reduces $d$-dimensional orthogonal range reporting queries to a small number of $(d-1)$-dimensional queries. Range trees provide a general method to solve $d$-dimensional range reporting queries for any constant dimension $d$. In this paper we use this data structure to reduce four-dimensional 5-sided reporting queries to three-dimensional 3-sided  queries. A range tree for a set $S$ is a balanced tree that holds the points of $S$ in the leaf nodes, sorted by their $z'$-coordinate. We associated a set $S(u)$ with every internal node $u$. $S(u)$ contains all points $p$ that are stored in the leaf descendants of $u$.  We assume that each internal node of $T$ has $\rho=\log^{\eps} n$ children. Thus every point is stored in $O(\log n/\log \log n)$ sets $S(u)$.  We keep a data structure that supports three-dimensional reporting queries on $S(u)$ for every node $u$ of the range tree. 

Consider a  query $Q=Q_3\times [a,b]$, where $Q_3$ denotes a 3-sided three-dimensional query range. Let $\ell_a$ be the rightmost leaf that holds some $z'$-coordinate $a'< a$ and let $\ell_b$ be the leftmost leaf that holds some $b'> b$. Let $v_{ab}$ denote the lowest common ancestor of $\ell_a$ and $\ell_b$.  We denote by  $\pi_a$ (resp.\ $\pi_b$) the set of nodes on the path from $\ell_a$ ($\ell_b$) to $v_{ab}$, excluding the node $v_{ab}$.  We will say that $u$ is a left (right) sibling of $v$ iff $u$ and $v$ have the same parent node and $u$ is to the left (respectively, to the right) of $v$.  The set $\pi'_a$ consists of all nodes $u$ that have some left sibling $v\in \pi_a$ and $\pi'_b$ consists of all nodes $u$ that have a right sibling $v\in \pi_b$.  The set $\pi''_a$ ($\pi''_b$) consists of all nodes in $\pi'_a$ (resp.\ in $\pi'_b$) that are not children of $v_{ab}$. The set $\pi'_{ab}$ consists of all children of $v_{ab}$ that are in $\pi'_a\cap \pi'_b$. For any point $p\in S$, $a\le p.z'\le b$ iff $p\in S(u)$ for some $u$ in $\pi''_a\cup \pi''_b\cup \pi'_{ab}$.  Nodes $u\in \pi''_a\cup\pi''_b\cup\pi'_{ab}$ are called canonical nodes for the range $[a,b]$. See an example on Fig.~\ref{fig:rangetree}. In order to answer a four-dimensional query $Q$ we visit every  canonical node $u$ and report all points $p\in S(u)\cap Q_3$.
\begin{figure}[tb]
  \centering
  \includegraphics[width=.45\textwidth]{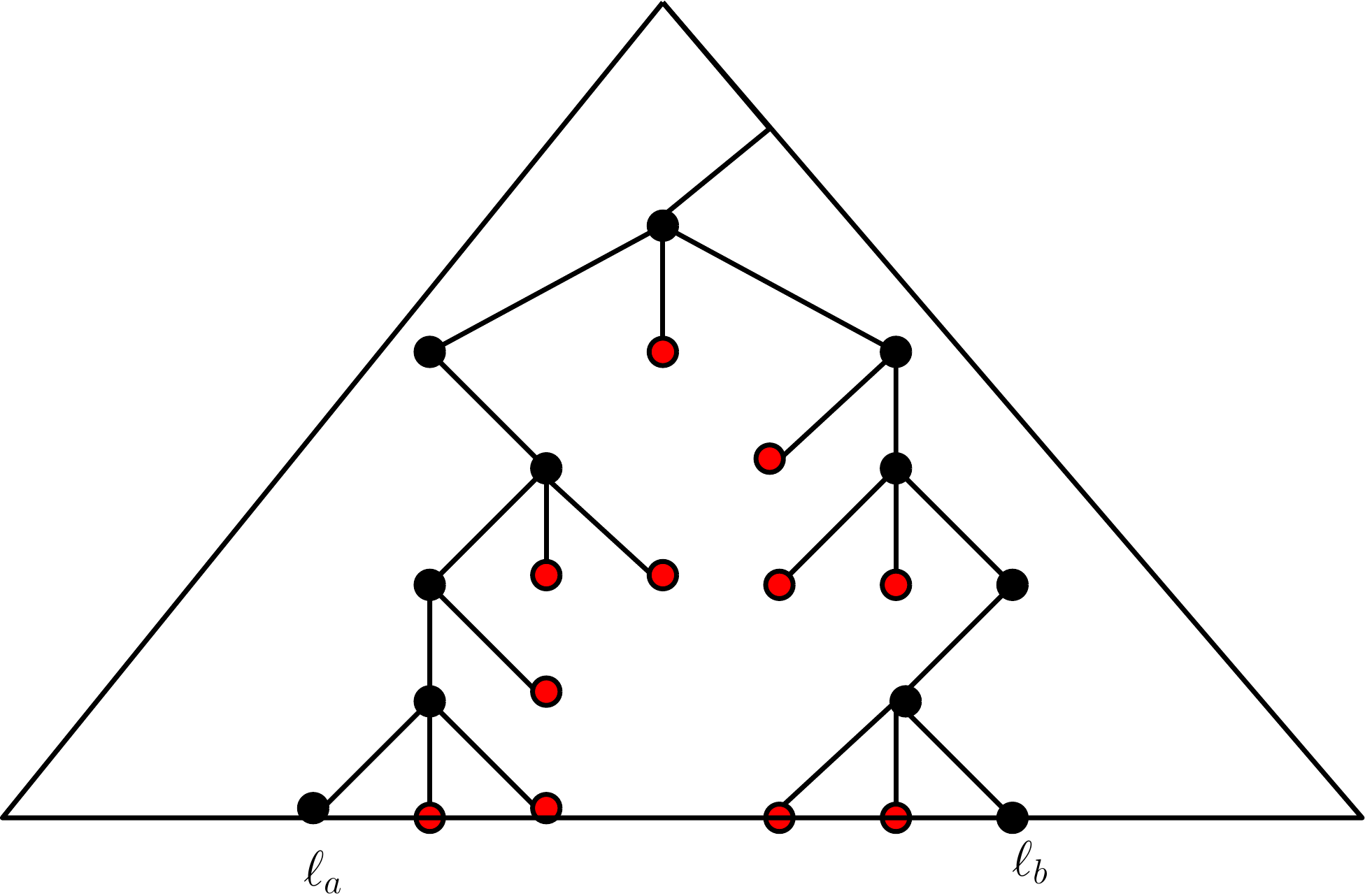}
  \caption{Example of a query in a range tree with node degree $\rho=3$. Canonical nodes are shown in red. Only nodes of $\pi_a\cup \pi_b$ and $\pi'_a\cup \pi'_b$ are shown. }
  \label{fig:rangetree}
\end{figure}

Thus we can answer a four-dimensional 5-sided query by answering  $O(\rho\cdot \log n/\log\log n)$  three-dimensional 3-sided queries in canonical nodes.

\subparagraph*{Rank Space.}
Let $E$ be a set of numbers. The \emph{rank} of a number  $x$ in $E$ is the number of elements in $E$ that do not exceed $x$: $\ra(x,E)=|\{\,e\in E\,|\, e\le x\,\}|$. Let $\pred(x,E)=\max\{\,e\in E\,|\,e\le x\,\}$ and $\ssucc(x,E)=\min\{\,e\in E\,|\, e\ge x\,\}$. An element $e\in E$ is in the range $[a,b]$, $a\le e\le b$, iff its rank satisfies the inequality $a'\le \ra(e,E) \le b'$ where $a'=\ra(\ssucc(a,E),E)$ and $b'=\ra(b,E)$.  Hence we can report all $e\in E$ satisfying $a\le e\le b$ by finding all elements $e$  satisfying $a'\le \ra(e,E)\le b'$. 

The same approach can be also extended to multi-dimensional range reporting. A three-dimensional transformation is implemented as follows. We say that a three-dimensional point $p\in S$ is reduced to rank space (or $p$ is in the rank space) if  each coordinate of $p$ is replaced by its rank in the set of corresponding  coordinates. That is,  each point $p=(p.x,p.y,p.z)$ is replaced with $\xi(p)=(\ra(p.x,S_x),\ra(p.y,S.y), \ra(p.z,S_z))$, where $S_x$, $S_y$, and $S_z$ denote the sets of $x$-, $y$-, and $z$-coordinates of points in $S$. For a point $p\in S$ we have 
\[ p\in [a,b]\times [c,d]\times [e,f] \Leftrightarrow \xi(p)\in [a',b']\times [c',d']\times [e',f']\]  where $a'=\ra(\ssucc(a,S_x),S_x)$, $c'=\ra(\ssucc(c,S_y),S_y)$, $e'=\ra(\ssucc(e,S_z),S_z)$, $b'=\ra(b,S_x)$, $d'=\ra(d,S_y)$, and $f'=\ra(f,S_z)$. Thus we can reduce three-dimensional queries on a set $S$ to three-dimensional queries on a  set $\{\,\xi(p)\,|\,p\in S\,\}$.  Suppose that  we store the set $\xi(S)$ in a data structure that answers queries in time $t(n)$ and uses space $s(n)$. Suppose that we also keep data structures that answers predecessor queries on $S_x$, $S_y$, and $S_z$.  Then we can answer orthogonal range reporting queries on $S$  in time $t(n)+O(t_{\pred}(n))$ using space $s(n)+O(s_{\pred}(n))$. Here $s_{\pred}(n)$ and $t_{\pred}(n)$ are the space usage and query time of the predecessor data structure.  Additionally we need a look-up table that computes $\xi^{-1}(p)$, i.e., computes the coordinates of a point $p$ from its coordinates in the rank space. As we will show later, in some situations this table is not necessary. Summing up, reduction to rank space enables us to reduce the range reporting problem on a set $S\subset \mathbb{R}^3$  to a special case when all point coordinates are positive integers bounded by $|S|$.  

The same rank reduction technique can be applied to range reporting in any constant dimension $d$. In the rest of this paper we will assume for simplicity that all points of $S$ are in the rank space.


\section{Five-Sided Range Reporting  in Linear Space}
\label{sec:5sid4dim}
\subparagraph*{Base Structure.} 
We keep all points in  a range tree  that is built on the fourth coordinate. Each tree node has $\rho=\log^{\eps} n$ children; thus the tree height is $O(\log n/\log\log n)$. Let $S(u)$ denote the set of points assigned to a node $u$.  To simplify the notation, we will not distinguish between points in $S(u)$ and their projections onto $(x,y,z)$-space. We will say for example that a point $p$ is in a range $Q=\halfr{a}\times\halfr{b}\times\halfr{c}$ if the projection of $p$ onto $(x,y,z)$-space is in $Q$.

Since we aim for a linear-space data structure, we cannot store sets $S(u)$ in the nodes of the range tree.  We keep a $t_0$-shallow cutting $\cC(u)$ of $S(u)$ where $t_0=\log^2 n$. For every cell $C_i(u)$ of the shallow cutting  we store all points from $S(u)\cap C_i$ in a data structure supporting three-dimensional dominance queries.  We do not store the original (real) coordinates of points\footnote{To avoid clumsy notation, we will sometimes omit the node specification when the node is clear from the context. Thus we will sometimes write $C_i$ instead of $C_i(u)$ and $\cC$ instead of $\cC(u)$. The same simplification will be used for other shallow cuttings.} in $C_i$. Instead we keep coordinates in the rank space of $S(u)\cap C_i$. Since $S(u)\cap C_i$ contains $O(\log^2 n)$ points, we need only $O(\log \log n)$ bits per point to answer three-dimensional dominance queries in the rank space. 

We can answer a 5-sided query $\halfr{q_x}\times \halfr{q_y}\times\halfr{q_z}\times [z'_l,z'_r]$ by visiting all canonical nodes that cover the range $[z'_l,z'_r]$. In every visited node we answer a three-dimensional dominance query, i.e., report all points dominated by $q_3=(q_x,q_y,q_z)$ in two steps: first, we search for some cell $C_i(u)$ that contains $q_3$. If such a cell $C_i(u)$ exists, then we answer the dominance query in the rank space of $S(u)\cap C_i(u)$ in $O(1)$ time per point. We describe the data structure for dominance queries on $t_0$ rank-reduced points in Section~\ref{sec:small}. 

We must  address several issues in order to obtain a working solution: How can we transform a three-dimensional query to the rank space of $C_i(u)$? A data structure for cell $C_i(u)$ reports points in the rank space of $C_i(u)$. How can we obtain the original point coordinates? Finally  how do we answer a query on $S(u)$ if $q_3$ is not contained in any cell $C_i(u)$? First, we will explain how to decode points from a cell $C(u)$ and obtain their original coordinates. We also show how to transform a query to the rank space of a cell.  Next we will describe a complete procedure for answering a query. Finally we will improve the query time and achieve  the main result of this section. 

\subparagraph*{Decoding Points.}
This is the crucial component of our construction. We will need additional structures to obtain the original coordinates of points from $\cC(u)$.  To this end  we keep an  additional $(4t_0$)-shallow cutting $\cC'(u)$ in every node of the range tree.  For each cell $C'_i(u)$ of $\cC'(u)$  we create a separate $(2t_0)$-shallow cutting  of $S(u)\cap C'_i(u)$, called $\cD_i(u)$. 

\begin{lemma}
\label{lemma:contain}
Let $\cA$ be an $f$-shallow cutting for a set $S$ and let $\cB$ be an $(f')$-shallow cutting for a set $S'\subseteq S$ so that $f'\ge 2f$. Every cell $A_i$ of $\cA$ is contained in some cell $B_j$ of $\cB$.     
\end{lemma}
\begin{proof}
  Consider an apex point $a_i$ of $A_i$ (i.e., the point with maximum $x$- $y$-, and $z$-coordinates in $A_i$). The level of $a_i$ in $S$ is at most $2f$ by definition of a shallow cutting. 
Since $S'\subseteq S$, the level of $a_i$ in $S'$ does not exceed $2f$. Hence there exists a cell $B_j$ of $\cB$ that contains $a_i$. The apex $b_j$  of $B_j$ dominates $a_i$. Hence $b_j$ also dominates all points from $A_i$ and $B_j$ contains $A_i$. 
\end{proof}
Lemma~\ref{lemma:contain} will be extensively used in our decoding procedure.
We will say that a point $p$ in $C'_i(u)$ is \emph{interesting} if  $p$ is contained in some $C_k(w)$, where $w$ is an ancestor of $u$. Each interesting point $p\in S(u)$  can be uniquely represented by  (a) a cell $C'_i(u)$ of $\cC'(u)$ that contains $p$ and (b) coordinates of $p$ in the rank space of $C'_i(u)$. 
The following relationship between shallow cuttings provides the  key insight. 



\begin{lemma}
  \label{lemma:cuttings}
(i) Every cell $C_i(u)$ of $\cC(u)$ is contained in some cell $C'_j(u)$ of $\cC'(u)$ \\
(ii) Let $u_r$ be a child of an internal node $u$. Every cell $D_{ij}(u)$ of every $\cD_i(u)$ is contained in some cell $C'_k(u_r)$ of $\cC'(u_r)$. \\
(iii)   Every interesting point from $C'_i(u)$  is stored in some cell $D_{ij}(u)$ of $\cD_i(u)$.     
\end{lemma}
\begin{proof}
(i) Immediately follows from Lemma~\ref{lemma:contain}.
(ii) Consider the apex point $p_a$  of $D_{ij}(u)$. 
The point $p_a$ dominates at most $4t_0$ points from $S(u)$ and at most $4t_0$ points from $S(u_r)$ because $S(u_r)\subset S(u)$. Hence both $p_a$ and $D_{ij}(u)$ are contained in some cell of $C'(u_r)$.
(iii) Suppose that a point $p\in S(w)$ is stored in some cell $C_k(w)$ of  $\cC(w)$ where $w$ is an ancestor of $u$. The point $p$ dominates at most $2t_0$ points from $S(w)$. 
Since $S(u)$ is a subset of $S(w)$, $p$ dominates at most $2t_0$ points in $S(u)$. Hence $p$ is contained in some cell $C'_i(u)$ of the shallow cutting  $\cC'(u)$.  Every point of $C'_i(u)\cap S(u)$ that dominates at most $2t_0$ points of $S(u)$ is contained in some cell $D_{ij}(u)$ of $\cD_{i}(u)$. 
\end{proof}

Consider an arbitrary point $p$ in a cell $C_i(u)$ of  $\cC(u)$. Our decoding procedure finds a representation of $p$ in $\cC'(u)$. That is, we find the cell $C'_j(u)$ of $\cC'(u)$, such that $p\in C'_j(u)$,   
and the rank of $p$ in $C'_j(u)$. The key observation is that $C_i(u)$ is contained in some $C'_j(u)$ (Lemma~\ref{lemma:cuttings}, item (i)) Therefore we need to store  a pointer to $C'_j(u)$ only once for all points $p$ in $C_i(u)$. For every $p$ from $C_i(u)$, we can store  its rank  in $C'_j(u)$ using $O(\log\log n)$ bits.   
Next, our decoding procedure  moves from a node $u$ to its child $u_f$, such that $p\in S(u_f)$, and computes a representation of $p$ in $\cC'(u_f)$. This is done in two steps: 
First we examine the shallow cutting $\cD_j(u)$ and find the cell $D_{jl}(u)$ that contains $p$. By Lemma~\ref{lemma:cuttings}, item (iii), such a cell always exists.  The shallow cutting $\cD_{jl}(u)$ consists of $O(1)$ cells. Therefore we can store, for any interesting point $p$,  the cell $D_{jl}(u)$ containing $p$ and the rank of $p$ in $D_{jl}(u)$ using $O(\log\log n)$ bits.
Then we move from $D_{jl}(u)$ to $\cC(u_f)$: by Lemma~\ref{lemma:cuttings}, item (ii), $D_{jl}(u)$ is contained in some $C'_k(u_f)$. Thus we need to store the pointer to $C'_k(u_f)$ only once for 
all interesting points $p$ in $D_{jl}(u)$. We can store the rank of $p$ in $C'_k(u_f)$ using $O(\log \log n)$ bits.  When the representation of $p$ in $\cC'(u_f)$ is known,  we move to the child $u_d$ of $u_f$, such that $p\in S(u_d)$ and compute a  representation of $p$ in $\cC'(u_d)$. We continue in the same way until a leaf node is reached. Every leaf node $\ell$ contains original (real) coordinates of points in $S(\ell)$. Hence we can obtain the coordinates of $p$ when a leaf is reached.  Summing up, shallow cuttings $\cC'(u)$ and $\cD_i(u)$ allow us  to move from a node $u$  to a  child of $u$ using only $O(\log \log n)$ additional bits per point. A more detailed description of auxiliary data structures needed for decoding  is given in the next paragraph.

For each cell $C_i(u)$ of $\cC(u)$  we keep a pointer to the cell $C'_{\cont(i)}(u)$ of $\cC'(u)$ that contains $C_i(u)$. For every cell $D_{ij}(u)\in \cD_{i}(u)$ and for each child $u_r$ of $u$, we keep a pointer to the cell $C_{\down(i,j,r)}'(u_r)\in \cC'(u_r)$, such that $C_{\down(i,j,r)}'(u_r)$ contains  $D_{ij}(u)$. We can identify a point in each cell of a shallow cutting $\cC(u)$ (resp. $\cC'(u)$ or $\cD'_i(u)$) with $O(\log\log n)$ bits because each cell contains a poly-logarithmic number of points. The $x$-rank of a point in a cell will be used as its identifier.  We keep a mapping from points in a cell $C_i$ to the corresponding points in a containing cell $C'_{\cont(i)}$. The array $F_X(C_i)$ maps $x$-ranks of points in $C_i$ to their $x$-ranks in $C'_{\cont(i)}$:  if the $x$-rank of a point $p\in C_i$ is equal to $f$, then $F_X[f]=g$ where $g$ is the $x$-rank of $p$ in $C'_{\cont(i)}$. The array $F''_{X,r}$ for a cell $D_{ij}(u)$  and a child $u_r$ of $u$ maps $x$-ranks of points in $D_{ij}(u)$ to their $x$-ranks in $C'_{\down(i,j,r)}$. If the $x$-rank of a point $p\in C_i$ is equal to $f$, then $F''_{X,r}[f]=g$ where $g$ is the $x$-rank of $p$ in $C'_{\down(i,j,r)}$.  We also keep a mapping from $C'_i(u)$ to cells of $\cD_i(u)$:   for every point $p\in C_i(u)$ we store the cell $D_{ij}$ that contains $p$  and the $x$-rank of $p$ in $C_{ij}$ (or $NULL$ if $p$ is not in $\cC_{ij}$). For every point $p$ in each cell $D_{ij}(u)$ of $\cC'_i(u)$, we store the index of the child $u_r$ such that $p\in S(u_r)$. Our method requires $O(\log\log n)$ bits per point. Each pointer from $C_i(u)$ to $C'_{\cont(i)}(u)$ and from $D_{ij}(u)$ to $C'_{\down(i,j,r)}(u)$ consumes $O(\log n)$ bits. We store $O(\log^{2\eps}n)$ pointers per cell and there are $O(n/(\log n\log\log n))$ cells in all shallow cuttings of the range tree. Hence the total space used by all pointers is $O(n\log^{2\eps}n)$ bits.

\begin{lemma}
\label{lemma:down}
  For any interesting point $p$ in a cell $C'_i(u)$, we can  find the representation of $p$ in $C'_i(u_f)$, where $u_f$ is the child of $u$ that contains $p$. 
\end{lemma}
\begin{proof}
  First we identify the cell $D_{ij}(u)$ of $\cD_i(u)$ that contains $p$ and compute the $x$-rank of $p$ in $D_{ij}(u)$. Since $p$ is interesting, such a cell exists. Then we use the array $F''_{X,k}$ of this cell and find the $x$-rank of $p$ in the cell $C'_{\down(i,j,k)}$. 
\end{proof}

For any point from $C_i(u)$ we can obtain its position in some cell $C'_{\cont(i)}(u)$ in $O(1)$ time. Then we can move down and obtain its representation in a child of $u$ in $O(1)$ time. We can access the original coordinates of $p$  when a leaf node is reached. Thus we can ``decode'' a point $p$ if we know its position in a cell $C_i(u)$ in $O(\log n/\log\log n)$ time. 

We can reduce a three-dimensional query $\halfr{a}\times\halfr{b}\times\halfr{c}$ to the rank space of a cell by binary search. Let $X(C_i)$ denote the list of points  in a cell $C_i$ sorted by $x$-coordinates. To compare $a$ with the $x$-coordinate of $X(C_i)[g]$ for some index $g$, we decode the point $p=X(C_i)[g]$ as explained above. Hence we can find the predecessor of $a$ in $X(C_i)$ by binary search in $O(\log \log n)$ time. We can find the predecessor of $b$ in $Y(C_i)$ and the predecessor of $c$ in $Z(C_i)$ using the same procedure, where $Y(C_i)$ and $Z(C_i)$ are the lists of points in $C_i$ sorted by their $y$- and $z$-coordinates respectively.

\subparagraph*{Queries.}
Consider a four-dimensional 5-sided query $\halfr{a}\times\halfr{b}\times\halfr{c}\times [z'_l,z'_r]$. We visit all canonical nodes that cover the range $[z'_l,z'_r]$. In every visited  node we answer a three-dimensional query using the following procedure. We find a cell $C_i(u)$ that contains $p$. We transform $\halfr{a}\times \halfr{b}\times \halfr{c}$ to the rank space of $C_i(u)$ and answer the transformed query on $C_i(u)\cap S(u)$. Every reported point is decoded using the procedure described  above.  If there is no cell $C_i(u)$ that contains $p$, then $p$ dominates at least $t_0$ points from $S(u)$. In this case we visit all children of $u$ and recursively answer three-dimensional dominance query in each child using the same procedure. 

We need $O(\log \log n)$ time to find the cell $C_i(u)$ or determine that $C_i(u)$ does not exist. To answer a query on $C_i(u)$ we need $O(\log n)$ time (ignoring the time to report points, but taking into account the time that we need to transform a query to the rank space of $C_i(u)$). Thus the total time spent in a node $u$ is $O(\log n)$. The time spent in descendants of $u$ can be estimated as follows. Let $T_u$ be the subtree of the range tree induced by $u$ and its visited descendants. Let $T_u'$ denote the subtree of $T_u$ obtained by removing all leaves of $T_u$. Every leaf of $T'_u$ is an internal node of $T_u$. Hence we report at least $t_0$ points for every leaf in $T'_u$.  The height of $T'_u$ is bounded 
by $O(\log n/\log\log n)$. Let $l_u$ denote the number of leaves in $T'_u$. The total number of nodes in $T'_u$ is bounded by $O(l_u\log n/\log \log n)$. Every node of $T'_u$ has $\rho$ children.  Hence the total number of nodes in $T_u$ does not exceed $O(l_u(\log^{1+\eps}n/\log\log n))$. The time spent in all nodes of $T_u$ can be bounded by $O(l_u\log^{2+\eps} n)$ (again, ignoring the time to decode and report points). When we visit descendants of $u$ we report at least $k_u=\Omega(l_u\cdot t_0)$ points and each point is decoded in $O(\log n/\log\log n)$ time. The total time spent in descendants of $u$ is $O(l_u\log^{2+\eps}n+  k_u(\log n/\log\log n)=O(k_u(\log n/\log\log n))$. The time spent in all canonical nodes and their descendants is $O(\log^{2+\eps}n +k(\log n/\log\log n))$.


\subparagraph*{Faster Decoding.}
We can  speed-up the decoding procedure and thus the overall query time without increasing the asymptotic space usage. Our approach is very similar to the method used in compact two-dimensional range trees~\cite{Chazelle88,Nekrich09,ChanLP11}.  All nodes in the range tree are classified according to their depth. A node $u$ is an $i$-node if the depth $h_u$ of $u$ divides $\rho^i$ where $\rho=\log^{\eps} n$, $h_u=x\cdot \rho^i$ for some $i\ge 0$, but $h_u$ does not divide $\rho^j$ for $j>i$. 
We keep an additional $4t_i$-shallow cutting $\cC^i$ in every $i$-node $u$ where $t_i=\rho^i\cdot \log^2 n$. As before for each cell $C^j_i$ of $\cC^j$ we construct a $2t_j$-shallow cutting $\cD^j_{i}$. Let an $i$-descendant of a node $u$ denote the highest $i$-node $v$ that is a descendant of $u$. If a node $u$ is an $i$-node, then it has $\rho^i$ $i$-descendants. For every cell $D^j_{ik}$ 
of each $\cD^j_i$ and for every $j$-descendant $u_l$ of $u$, we keep the index $r=\down(j,i,l)$ of the cell $C^j_r(u_l)$ that contains $D^j_{ik}(u)$. For each point in $D^j_{ik}(u)$ we keep the index $l$ of the $i$-descendant that contains $p$ and the $x$-rank of $p$ in $C^j_r(u_l)$ where $r=\down(j,i,l)$. 

Using these additional shallow cuttings, we can reduce the decoding time to $O(\log^{\eps} n)$. To decode a point $p$ in $S(u)$ we move down from a node $u$ to its child $u_{0,1}$  and find a representation of $p$ in $\cC'(u_{0,1})$. Then we move to the child $u_{0,2}$ of $u_{0,1}$ and continue in the same manner until a $1$-node $u_{1,1}$ is reached. Next we move from $u_{1,1}$ to its $1$-descendant $u_{1,2}$, then to a $1$-descendant $u_{1,3}$ of $u_{1,2}$, and so on until a $2$-node is reached. During the $j$-th iteration we move down along  a sequence of $j$-nodes until a $(j+1)$-node  or a leaf node is reached. During each iteration we visit $O(\log^{\eps} n)$ nodes and spend $O(1)$ time in every node. There are at most $\log_{\rho}\log n=O(1/\eps)$ iterations.  Hence the decoding time for a point is $O(\log^{\eps} n)$. The total query time is reduced to $O(\log^{1+2\eps}n+k\log^{\eps}n)$. If we replace $\eps$ with $\eps/2$ in the above proof, we obtain our first result.
\begin{theorem}
  \label{theor:linspace}
There exists a linear-space data structure that answers four-dimensional 5-sided reporting queries in $O(\log^{1+\eps} n+ k\log^{\eps} n)$ time and four-dimensional 5-sided emptiness queries in  $O(\log^{1+\eps} n)$ time.
\end{theorem}

\section{Faster Queries using More Space}
\label{sec:decodingfast}
In this section we will show how to reduce the decoding time to $O(1)$ per point by increasing the space usage. We make several modifications in the basic construction of Section~\ref{sec:5sid4dim}. 

For any $i$ and $j$ such that $1\le i\le j \le \rho$ and for any internal node $u$ of the range tree, we store the set $S(u,i,j)$. $S(u,i,j)$  is the union of sets $S(u_i)$, $S(u_{i+1})$, $\ldots$, $S(u_j)$. For every set $S(u,i,j)$ we construct a $t_0$-shallow cutting $\cC(u,i,j)$. 
For each cell $C_l$ of $\cC(u,i,j)$ we store a three-dimensional data structure that keeps points from $C_l\cap S(u,i,j)$ in the rank space and answers three-dimensional dominance queries in $O(k+1)$ time. 

The decoding procedure is implemented in the same way as in Section~\ref{sec:5sid4dim}, but with different parameter values. Recall that a node $u$ is an $i$-node for some $i\ge 0$ if the depth of $u$ divides $\rho^i$ but does not divide $\rho^{i+1}$. We keep an
additional $4t_{i+1}$-shallow cutting $\cC^i(u,l,r)$ for every $i$-node $u$ and every pair $1\le l\le r\le \rho$. For every cell  $C_s$ of $\cC^i(u,l,r)$ we keep a $2t_{i+1}$-shallow cutting $\cD_s$. Consider a cell $D_g$ of $\cD_s$.  For every $(i+1)$-descendant $v$ of $u$ and for every pair $l$, $r$ satisfying $1\le l\le r\le\rho$, we keep the index $x=\down(D_g,v,l,r)$  such that the cell $C_x$ of $\cC^{i+1}(v,l,r)$ contains $D_g$.  We also store a mapping from $\cC(u,l,r)$ to $\cC^i(u,l,r)$ for every $i$-node $u$. That is, for every cell $C_f$ of $\cC(u,l,r)$ we keep the index $g=\cont(f)$, such that the cell $C_g\in \cC^i(u,l,r)$ contains $C_f$; for every  point $p\in S(u)\cap C_f$ we keep its identifier in $C_{\cont(f)}$. For every cell $C_g$ of $\cC^i(u,l,r)$ we keep a mapping from $C_g$ to $\cD_g$.  That is, for every point $p$ in $C_g\cap S(u,l,r)$ we store the cell $D_s$ of $\cD_g$ that contains $p$ and the identifier of $p$ in $D_s$. Finally we also store a mapping from every cell $D_s$ of each $\cD_g$ to shallow cuttings in $(i+1)$-descendants of $u$. For every point $p\in D_s\cap S(u,l,r)$ we store (i) the $i$-descendant $v$ of $u$ such that $p\in S(v)$ and (ii) the identifier of $p$ in $C_x$ where $x=\down(D_s,v,l,r)$.

Our modified data structure uses $O(n\log^{3\eps}n)$ words of space. The representation of a point in  $\cC(u,i,j)$ takes $O(\log \log n)$ bits per point  and every point is stored in $O(\log^{1+2\eps}n/\log \log n)$ shallow cuttings $\cC(u,i,j)$.  The mapping from $\cC(u,l,r)$ to $\cC^i(u,l,r)$ in an $i$-node $u$ takes $O(\log^{(i+1)\eps}n)$ bits per point.  We also need $O(\log^{(i+1)\eps}n)$ bits per point to store the mapping from a cell $C_g$ of $\cC^i(u,l,r)$ to $\cD_g$. The mapping from 
a cell $D_s$ of $\cD_g$ to shallow cuttings in $(i+1)$-descendants of $u$ consumes the same space. The total number of points in all $S(u)$ where $u$ is an $i$-node is  $O(n\log^{1-i\eps}n)$. The total number of points in all $S(u,l,r)$ where $u$ is an $i$-node and $1\le l\le r\le \rho$ is  $O(n\log^{1+(2-i)\eps}n)$.  Hence the total space used by all mappings in all $i$-nodes is $O(n\log^{1+3\eps} n)$ bits or $O(n\log^{3\eps}n)$ words of $\log n$ bits.   

Every point  $p$ in $C_i\cap S(u,l,r)$, where $C_i$ is a cell of $\cC(u,l,r)$,  can be decoded in $O(1)$ time. Suppose that $u$ is a $j$-node. Using the mapping from $C_i$ to  $\cC^j(u,l,r)$, we can find the representation of $p$ in $\cC^j(u,l,r)$, i.e., a cell $C_s$ that contains $p$ and the identifier of $p$ in $C_s$. If we know the identifier of $p$ in $C_s$, we can find the representation of $p$ in $\cD_s$. Using the mapping from a cell of $\cD_s$ to $(j+1)$-descendants of $u$, we can compute  the representation of $p$ in a cell $C_v$ of $\cC^{j+1}(v,l',r')$, where $v$ is a direct $(j+1)$-descendant of $u$. Thus we moved from a $j$-node to its $(j+1)$-descendant in $O(1)$ time. We continue in the same way and move to  a $(j+2)$-descendant of $u$, then a $(j+3)$-descendant of $u$, and so on. After at most $(1/\eps)=O(1)$ iterations, we reach a leaf node and obtain the original coordinates of $p$.  

We can   translate a query $\halfr{a}\times\halfr{b}\times\halfr{c}$ into the rank space of a cell $C_i$ in constant time. 
Let $X(C_i)$ denote the list of $x$-coordinates of $S(u,l,r)\cap C_i$. We keep $X(C_i)$ in the compact trie data structure of 
~\cite{GrossiORR09}. This data structure requires $O(\log\log n)$ bits per point.  Elements of $X(C_i)$ are not stored in the compact trie; we only store some auxiliary information using $O(\log \log n)$ bits per element. Compact trie supports predecessor queries on $X(C_i)$ in $O(1)$ time, but the search procedure must access $O(1)$ elements  of $X(C_i)$.  Since we can decode a point from $C_i$ in $O(1)$ time, we can also access an element of $X(C_i)$ in $O(1)$ time. Hence, we can compute the predecessor of $a$ in $X(C_i)$ (and its rank) in $O(1)$ time. We can translate $b$ and $c$ to the rank space in the same way.

\subparagraph*{Queries.}
Consider a four-dimensional 5-sided query $\halfr{a}\times\halfr{b}\times\halfr{c}\times [z'_1,z'_2]$.  There are 
$O(\log n/\log\log n)$ canonical sets $S(u_i,l_i,r_i)$, such that $p.z\in [z'_1,z'_2]$ iff $p\in S(u_i,l_i,r_i)$ for some $i$. Canonical sets can be found as follows. Let $\ell_1$ be the leaf that holds the largest $l_1 < z_1'$ and $\ell_2$ be the leaf that holds the smallest $l_2>z'_2$. Let $v$ denote the lowest common ancestor of $\ell_1$ and $\ell_2$. Let $\pi_1$ denote the path from $\ell_1$ to $v$ (excluding $v$) and let $\pi_2$ denote the path from $\ell_2$ to $v$ (excluding $v$). For each node $u\in \pi_2$, we consider a canonical set $S(u,l,r)$ such that $u_l$, $\ldots$, $u_r$ are left siblings of some node $u_{r+1}\in \pi_2$. For each node $u\in \pi_1$,  we consider a canonical set $S(u,l,r)$ such that $u_l$, $\ldots$, $u_r$ are right siblings of some node $u_{l-1}\in \pi_1$. Finally we consider the set $S(v,l,r)$ such that $v_l$, $\ldots$, $v_r$ have a left sibling on $\pi_1$ and a right sibling on $\pi_2$. The fourth coordinate of a point $p$ is in the interval $[z'_1,z'_2]$ iff $p$ is stored in one of the canonical sets described above. Hence we need to visit all canonical sets and answer a three-dimensional query $\halfr{a}\times\halfr{b}\times\halfr{c}$ in each set.

There are $O(\log n/\log \log n)$ canonical sets $S(u,l,r)$. Each canonical set is processed as follows. We find the cell $C_u$ of $\cC(u,l,r)$ that contains $q_3=(a,b,c)$. Then we translate $q_3$ into the rank space of $C_u\cap S(u,l,r)$ and answer the dominance query. Reported points are decoded in $O(1)$ time per point as explained above. We can also translate the query into the rank space of $C_u$ in $O(1)$ time.  If $q_3$ is not contained in any cell of $\cC(u,l,r)$, then $q_3$ dominates at least $\log^2 n$ points of $S(u,l,r)$.  We visit all children $u_i$ of $u$ for $l\le i \le r$ and recursively answer the dominance query in each child. Using the same arguments as in Section~\ref{sec:5sid4dim}, we can show that the total number of visited nodes does not exceed $O(k/\log^{\eps} n)$, where $k$ is the number of reported points. 

If we replace $\eps$ with $\eps/3$ in the above proof, we obtain the following result. 
\begin{theorem} 
  \label{lemma:faster}
There exists an $O(n\log^{\eps}n)$ space data structure that answers four-dimensional 5-sided reporting queries in $O(\log n+ k)$ time and four-dimensional 5-sided emptiness queries in  $O(\log n)$ time.
\end{theorem}
We can extend our result to support dominance queries (or any $(2d-3)$-sided queries) in $d\ge 4$ dimensions using standard techniques.
\begin{theorem}
  There exists an $O(n\log^{d-4+\eps}n)$ space data structure that supports $d$-dimensional dominance range reporting queries in 
  $O(\log^{d-3} n/(\log\log n)^{d-4}+k)$ time for any constant $d\ge 4$. \\
  There exists an $O(n\log^{d-4+\eps}n)$ space data structure that supports $d$-dimensional $(2d-3)$-sided range reporting queries in 
  $O(\log^{d-3} n/(\log\log n)^{d-4}+k)$ time for any constant $d\ge 4$. 
\end{theorem}


\section{Conclusions}
\label{sec:conclusion}
In this paper we described data structures with  linear and almost-linear space usage that answer four-dimensional range reporting queries in poly-logarithmic time provided that the query range is bounded on 5 sides. This scenario includes an important special case of dominance range reporting queries that was studied in a number of previous works~\cite{ChazelleE87,VengroffV96,Nekrich07,Afshani08,ChanLP11,Chan13}; for instance, the offline variant of four-dimensional  dominance reporting is used to solve the rectangle enclosure problem~\cite{ChanLP11,AfshaniCT14}. Our result immediately leads to better data structures in $d\ge 4$ dimensions. E.g., we can answer $d$-dimensional dominance range reporting queries in $O(n\log^{d-4+\eps}n)$ space and $O(\log\log n (\log n/\log\log n)^{d-3})$ time.  We  expect  that the methods of this paper can be applied to other geometric problems, such as the offline rectangle enclosure problem.

Our result  demonstrates that the space complexity of four-dimensional queries is determined by the number of sides, i.e., the number of inequalities that are needed to specify the query range. This raises the question about the space complexity of dominance range reporting in five dimensions. Is it possible to construct  a linear-space data structure that supports five-dimensional dominance range reporting queries in poly-logarithmic time?  

Compared to the fastest previous solution for the four-dimensional dominance range reporting problem~\cite{Chan13}, our method decreases the space usage by  $O(\log n)$ factor without increasing the query time.  However, there is still a small gap between the $O(\log n + k)$ query time, achieved by the fastest data structures, and the lower bound of $\Omega(\log n/\log\log n)$, proved in~\cite{Patrascu11}. Closing this gap is another interesting open problem.

\bibliographystyle{plain}
\bibliography{range-rep}
\appendix
\section{Dominance Reporting on a Small Set}
\label{sec:small}
In this section we explain, for completeness, how to answer three-dimensional dominance range reporting queries on a set $S$ of $t=\log^2 n$ points. We assume that points are stored in the rank space.  

 \begin{lemma}
   \label{lemma:small}
If a set $S$ contains $t=O(\log^2 n)$ points in the rank space of $S$, then we can keep $S$ in a data structure  that uses $O(t\log\log n)$ bits and answers three-dimensional dominance range reporting queries in $O(k)$ time. This data structure  uses  a universal look-up table of size $o(n)$.
 \end{lemma}
\begin{proof}
We can answer a query on a set $S'$ that contains at most $t'=(1/4)\log n/\log\log n$    points using a look-up table of size $o(n)$. Suppose that all points in $S'$ have positive integer coordinates bounded by $t'$. There are $2^{t'\log t'}$ combinatorially different sets $S'$.  For every instance of $S'$, we can ask $(t')^3$ different queries and the answer to each query consists of $O(t')$ points. Hence the total space needed to keep answers to all possible queries on all instances of $S'$ is $O(2^{(\log t')t'}(t')^5)=o(n)$ points. The general case when point coordinates are bounded by $d$ can be reduced to the case when point coordinates are bounded by $t'$ using reduction to rank space~\cite{GabowBT84,AlstrupBR00}; see Section~\ref{sec:prelim}. 

A query on $S$ can be reduced to $O(1)$ queries on sets that contain $O(t')$ points using the grid approach~\cite{ChanLP11,AlstrupBR00}. The set of points 
$S$ is divided into $\sqrt{t'}$ columns $C_i$ and $\sqrt{t'}$ rows $R_j$ so that every row and every column contains  $t/\sqrt{t'}$ points. 
The bottom set $S_b$ contains a meta-point $(i,j,z_{\min})$ iff the intersection of the $i$-th column and the $j$-th row is not empty: if  $R_j\cap C_i\not=\emptyset$ we store the point $(i,j,z_{ij})$ where $z_{ij}$ is the smallest $z$-coordinate of a point in $R_j\cap C_i$. Since $S_b$ contains at most $d'$ points, we can support queries on $S_b$ in $O(k)$ time. For each meta-point $(i,j)$ in $S_b$ we store the list of points $L_{ij}$ contained in the intersection of the $i$-th column and the $j$-th row, $L_{ij}=C_i\cap R_j\cap P$. Points in $L_{ij}$ are sorted in increasing order of their $z$-coordinates.  Every column $C_i$ and every row $R_j$ is recursively divided in the same manner: if $C_i$ or $R_j$ contains more than $t'$ points, we divide $C_i$ ($R_j$) into $\sqrt{t'}$ rows and $\sqrt{t'}$ columns of equal size and construct a data structure for the set $S_b$.  Since the number of points in a row (column) is decreased by a factor $\sqrt{t'}$ on every recursion level, our data structure has  at most five levels of recursion.

Consider a query $Q=\halfr{a}\times\halfr{b}\times \halfr{c}$. If $Q$ is contained in one column or one row, we answer the query using the data structure for that column/row. Otherwise we identify the row $R_u$ that contains $b$ and the column $C_r$ that contains $a$. Using the data structure on $S_b$, we find all 
meta-points $(i,j,z)$ satisfying $i<r$, $j<u$, and $z\le c$. 
For every found meta-point $(i,j,z)$ we output all points in $L_{ij}$ with $z$-coordinates not exceeding $c$.   
Then we recursively answer the query $Q$ on $R_u$ and $C_r$.  
Since there is a constant number of recursion levels, a query is answered in $O(k)$ time, where $k$ is the number of reported points.  Since each point is stored a constant number of times, the total space usage is $O(t\log\log n)$ bits. 
\end{proof}

This  result can be extended to any set with poly-logarithmic number of points, provided coordinates are in the rank space.

\end{document}